\documentclass[preprint]{elsarticle}
\usepackage{amsmath}
\usepackage{amssymb}
\usepackage{amsthm}
\usepackage[pdftitle={A lower bound on the tree-width of graphs with irrelevant vertices},pdfauthor={Philipp Klaus Krause, Isolde Adler}]{hyperref}

\newcommand{\coloneqq}{\mathrel{\mathop:}=}

\let\emptyset\varnothing
\let\setminus\smallsetminus
\graphicspath{{../../}}
\newtheorem{definition}{Definition}
\newtheorem{lemma}{Lemma}
\newtheorem{theorem}{Theorem}
\newtheorem{remark}{Remark}

\newcommand{\DPP}{\operatorname{DPP}}
\newcommand{\tw}{\operatorname{tw}}
\newcommand{\pw}{\operatorname{pw}}
\newcounter{claimcounter}

\begin{document}

\begin{frontmatter}
\title{A lower bound on the tree-width of graphs with irrelevant vertices\footnote{This version has been created by modifying a preprint submitted to Elsevier; modifications are minor and were done purely to work around lack of support for xelatex on arXiv}}
\author[leeds]{Isolde Adler}
\ead{I.M.Adler@leeds.ac.uk}
\author[freiburg]{Philipp Klaus Krause}
\ead{krauseph@informatik.uni-freiburg.de}
\address[leeds]{School of Computing, University of Leeds}
\address[freiburg]{Albert-Ludwigs-Universit{\"a}t Freiburg}
\begin{abstract}
	For their famous algorithm for the disjoint paths problem,
	Robertson and Seymour proved 
	that there is a function $f$ such that if the tree-width of a graph $G$ 
	with $k$ pairs of terminals is at least $f(k)$, then $G$ contains a solution-irrelevant vertex (Graph Minors. XXII., JCTB 2012).
	We give a single-exponential lower bound on $f$. This bound even holds for planar graphs.
\end{abstract}
\end{frontmatter}

\section{Introduction}
The {\sc Disjoint Paths Problem} is
one of the famous classical problems in the area of graph algorithms. 
Given a graph $G$, and $k$ pairs of terminals, $(s_1,t_1), \ldots, (s_k,t_k)$,
it asks whether $G$ contains $k$ vertex-disjoint paths $P_{1}, \ldots, P_{k}$ such that $P_{i}$ connects $s_{i}$ to $t_{i}$, (for $i=1, \ldots, k$). 
Karp proved that the problem is NP-hard in general~\cite{Karp75} and Lynch proved that it remains NP-hard on planar graphs~\cite{Lynch75}. 
Robertson and Seymour showed that it can be solved in time $g(k)\cdot |V(G)|^3$ for some computable function $g$, i.\,e.\ the problem is fixed-parameter tractable (and, in particular, solvable in polynomial time for fixed $k$).
For a recursive step in their algorithm 
((10.5) in~\cite{GMXIII}), they prove~\cite{GMXXII} that there is a function $f\colon\mathbb N\to\mathbb N$
such that if a graph $G$ with $k$ pairs of terminals has tree-width at least $f(k)$, then 
$G$ contains a vertex that is \emph{irrelevant} to the solution, i.\,e.\
$G$ contains a non-terminal vertex $v$ such that $G$ has a solution if and only if
the graph $G - v$ (with the same terminals) has a solution.

In this paper we give a lower bound on $f$, showing that $f(k) \geq 2^k$, even for planar graphs.
For this we construct a family of planar input graphs $(G_k)_{k \geq 2}$, each with $k$ pairs of terminals,
such that the tree-width of $G_k$ is $2^k - 1$, and every member of the family has a unique solution to the {\sc Disjoint Paths Problem}, 
where the paths of the solution use all vertices of the graph. Hence no vertex of $G_k$ is irrelevant. As a corollary, we obtain 
a lower bound of $2^k - 1$ on the tree-width of graphs having \emph{vital linkages}~(also called \emph{unique linkages})~\cite{GMXXI} with $k$ components.\footnote{This result appeared in the last section of a conference paper~\cite{AdlerKKLST11}. While the main focus of the paper~\cite{AdlerKKLST11} was a single exponential upper bound on $f$ on planar graphs, it only sketches the lower bound. Here we provide the full proof of the lower bound. A longer proof of the lower bound can also be found in the thesis~\cite{Krause2016}.}
Our result contrasts the polynomial upper bound in a related topological setting~\cite{Matousek2016}, where two systems of curves are untangled on a sphere with holes.

For planar graphs, an upper bound of $f(k) \leq 72\sqrt{2}k^{\frac{3}{2}} \cdot 2^k$ was given in~\cite{AdlerKKLST11}. An elementary proof for a bound of $f(k) \leq (72k \cdot 2^k - 72 \cdot 2^k + 18)\lceil \sqrt{2k + 1} \rceil$ was provided later~\cite{Krause2016} as well as a slightly improved bound of $f(k) \leq 26k \cdot 2 ^{\frac{3}{2}} \cdot 2^k$ requiring a slightly more involved proof~\cite{AdlerKKLST17}. Our lower bound shows that this is asymptotically optimal.
Recently, an explicit upper bound on $f$ on graphs of bounded genus~\cite{GeelenHR18} was found, then refined into one that is single exponential in $k$ and the genus~\cite{Mazoit13}. The exact order of growth of $f$ on general graphs is still unknown.

\section{Preliminaries}

Let $\mathbb N$ denote the set of all non-negative integers. For $k\in \mathbb N$, we let $[k] \coloneqq \{1, \ldots, k\}$. For a set $S$ we let $2^S$ denote the power set of $S$.
A \emph{graph} $G = (V, E)$ is a pair of a set of \emph{vertices} $V$ and a set of \emph{edges} $E \subseteq \{e\ |\ e \in 2^V, |e| = 2\}$, i.\,e.\ graphs are undirected and simple. For an edge $e = \{x, y\}$, the vertices $x$ and $y$ are called \emph{endpoints} of the edge $e$, and the edge is said to be between its endpoints. For a graph $G = (V, E)$ let $V(G) \coloneqq V$ and $E(G) \coloneqq E$.
Let $H$ and $G$ be graphs. The graph $H$ is a \emph{subgraph} of $G$ (denoted by $H \subseteq G$), if $V(H) \subseteq V(G)$ and $E(H) \subseteq E(G)$. For a set $X \subseteq V(G)$, the subgraph of $G$ \emph{induced by} $X$ is the graph $G[X] \coloneqq (X, \{e \in E(G)\ |\ e \subseteq X\})$ and we let $G - v \coloneqq G[V(G) \setminus \{v\}]$.

A \emph{path} $P$ in a graph $G = (V, E)$ is a sequence $n_0, \ldots, n_k \in V$ of pairwise distinct vertices of $G$, such that for every $i \in \{0, \ldots, k-1\}$ there is an edge $\{n_i, n_{i + 1}\} \in E$. The vertices $n_0$ and $n_k$ are called \emph{endpoints} of $P$. The path $P$ is called a path \emph{from} $n_0$ \emph{to} $n_k$ (i.\,e.\ paths are \emph{simple}). We  sometimes identify the path $P$ in $G$ with the subgraph $(\{n_0, \ldots, n_k\}, \{\{n_0, n_1\}, \ldots, \{n_{k - 1}, n_k\}\})$ of~$G$.
A graph $G$ is called \emph{connected}, if it has at least one vertex and for any two vertices $x, y \in V(G)$, there is a path from $x$ to $y$ in $G$. The inclusion-maximal connected subgraphs of a graph are called \emph{connected components} of the graph.
For $A, B \subseteq V(G)$, a set $S \subseteq V(G)$ separates $A$ from $B$, if there is no path from a vertex in $A$ to a vertex in $B$ in the subgraph of $G$ induced by $V(G) \setminus S$.
A \emph{tree} is a non-empty graph $T$, such that for any two vertices $x, y \in V(T)$ there is exactly one path from $x$ to $y$ in $T$.

\begin{figure}[h]
	\centering
	\scalebox{1.7}{\includegraphics{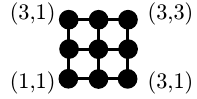}}
	\caption{$(3 \times 3)$-grid\label{ExGrid2}}
\end{figure}

	Let $m,n\in \mathbb N \setminus \{0\}$. The $(m \times n$)-\emph{grid}\index{grid} is a graph $H = (V, E)$ with $V \coloneqq [m] \times [n]$ and $E \coloneqq \{\{(y, x), (w, z)\}\ |\ (y, x) \in V, (w, z) \in V, |x - z| + |y - w| = 1\}$. In case of a square grid where $m = n$, we say that $n$ is the \emph{size}\index{size of a grid} of the grid. An edge $\{(y, x), (w, z)\}$ in the grid is called \emph{horizontal}\index{horizontal edge}, if $y = w$, and \emph{vertical}\index{vertical edge}, if $x = z$.
See Figure~\ref{ExGrid2} for the $(3 \times 3)$-grid.

A \emph{drawing}\index{drawing} of a graph $G$ is a representation of $G$ in the Euclidean plane $\mathbb R^2$, where vertices are represented by distinct points of $\mathbb R^2$ and edges by simple curves joining the points that correspond to their endpoints, such that the interior of every curve representing an edge does not contain points representing vertices.
A \emph{planar drawing}\index{planar drawing} (or \emph{embedding}\index{planar embedding}) is a drawing, where the interiors of any two curves representing distinct edges of $G$ are disjoint.
A graph $G$ is \emph{planar}\index{planar graph}, if $G$ has a planar drawing (See~\cite{MoharTh01} for more details on planar graphs).
A \emph{plane graph}\index{plane graph} is a planar graph $G$ together with a fixed embedding of $G$ in $\mathbb R^2$.
We will identify a plane graph with its image in $\mathbb R^2$.
Once we have fixed the embedding, we will also identify a planar graph with its image in $\mathbb R^2$.

\begin{definition}[Disjoint Paths Problem ($\DPP$)]
	Given a graph $G$ and $k$ pairs of \emph{terminals} $(s_{1},t_{1})\in V(G)^{2}, \ldots, (s_{k},t_{k})\in V(G)^{2}$, the {\sc Disjoint Paths Problem} is the problem of deciding whether $G$ contains $k$ vertex-disjoint paths $P_{1}, \ldots, P_{k}$ such that $P_{i}$ connects $s_{i}$ to $t_{i}$ (for $i \in [k]$).
	If such paths $P_1, \ldots, P_k$ exist, we refer to them as a \emph{solution}.
	We denote an \emph{instance}\index{instance of the disjoint paths problem} of $\DPP$ by $G, (s_1, t_1), (s_2, t_2), \ldots, (s_k, t_k)$.
\end{definition}

Let $G, (s_1, t_1), \ldots, (s_k, t_k)$ be an instance of $\DPP$. A non-terminal vertex $v\in V(G)$ is \emph{irrelevant}\index{irrelevant vertex}, if $G, (s_1, t_1), \ldots, (s_k, t_k)$ has a solution if and only if $G - v, (s_1, t_1), \ldots, (s_k, t_k)$ has a solution.

	A \emph{tree-decomposition}\index{tree-decomposition} of a graph $G$ is a pair  $(T, \chi)$, consisting of a tree $T$ and a mapping $\chi\colon V(T)\to 2^{V(G)}$, such that for each $v \in V(G)$ there exists $t \in V(T)$ with $v \in \chi(t)$, for each edge $e \in E(G)$ there exists a vertex $t \in V(T)$ with $e \subseteq \chi(t)$, and for each $v \in V(G)$ the set $\{t \in V(T) \mid v \in \chi(t)\}$ is connected in $T$.
	The \emph{width} of a tree-decomposition $(T,\chi)$ is
	\begin{displaymath}
		\operatorname{w}(T,\chi) \coloneqq \max\left\{\left|\chi(t)\right|-1\ \middle|\ t\in V(T)\right\}.
	\end{displaymath}
	If $T$ is a path, $(T, \chi)$ is also called a \emph{path-decomposition}\index{path-decomposition}.
	The \emph{tree-width of $G$} is
	\begin{displaymath}
		\tw(G) \coloneqq \min\left\{\operatorname{w}(T,\chi)\ \middle|\ (T,\chi) \text{ is a tree-decomposition of }G\right\}.
	\end{displaymath}

	The \emph{path-width of $G$} is
	\begin{displaymath}
		\pw(G) \coloneqq \min\left\{\operatorname{w}(T,\chi)\ \middle|\ (T,\chi) \text{ is a path-decomposition of }G\right\}.
	\end{displaymath}

Obviously, every graph $G$ satisfies $\pw(G)\geq \tw(G)$.
Every tree has tree-width at most $1$ and every path has path-width at most $1$. It is well known that the $(n \times n$)-grid has both tree-width and path-width $n$.
Moreover, if $H\subseteq G$, then $\tw(H)\leq \tw(G)$ and $\pw(H)\leq \pw(G)$.

\begin{theorem}[Robertson and Seymour \cite{GMXXII}]\label{thm:irr}
	There is a function $f\colon\mathbb N\to\mathbb N$ such that if $\tw(G) \geq f(k)$, 
	then $G, (s_1, t_1), \ldots, (s_k, t_k)$ has an irrelevant vertex\index{irrelevant vertex} (for any choice of 
	terminals $(s_1, t_1), \ldots, (s_k, t_k)$ in $G$).
\end{theorem}

	A \emph{linkage}\index{linkage} in a graph $G$ is a subgraph $L \subseteq G$, such that each connected component of $L$ is a path.
	The \emph{endpoints}\index{endpoints of a linkage} of a linkage $L$ are the endpoints of these paths, and the \emph{pattern}\index{pattern of a linkage} of $L$ is the matching on the endpoints induced by the paths, i.\,e.\ the pattern is the set
	\begin{displaymath}
		\left\{ \{s,t\}\ \middle|\ L\text{ has a connected component that is a path from $s$ to $t$}\right\}.
	\end{displaymath}
	A linkage $L$ in a graph $G$ is a \emph{vital linkage}\index{vital linkage} in $G$, if $V(L) = V(G)$ and there is no other linkage $L'\neq L$ in $G$ with the same pattern as $L$.

\begin{theorem}[Robertson and Seymour \cite{GMXXII}]\label{thm:vital}
	There are functions $g,h\colon\mathbb N\to\mathbb N$ such that if a graph $G$ has a vital linkage\index{vital linkage} with $k$ components then $\tw(G)\leq g(k)$ and $\pw(G)\leq h(k)$.
\end{theorem}

\section{The lower bound}

Our main result is the following.

\begin{theorem}\label{thm:lowerbound}
	Let $f,g,h\colon\mathbb N\to\mathbb N$ be as in Theorems~\ref{thm:irr} and~\ref{thm:vital}.
	Then $f(k) \geq 2^k$, $g(k)\geq 2^k-1$, and $h(k)\geq 2^k-1$. Moreover, this holds even if we consider planar graphs only.
\end{theorem}

\begin{figure}
	\centering
	\scalebox{1.7}{\includegraphics{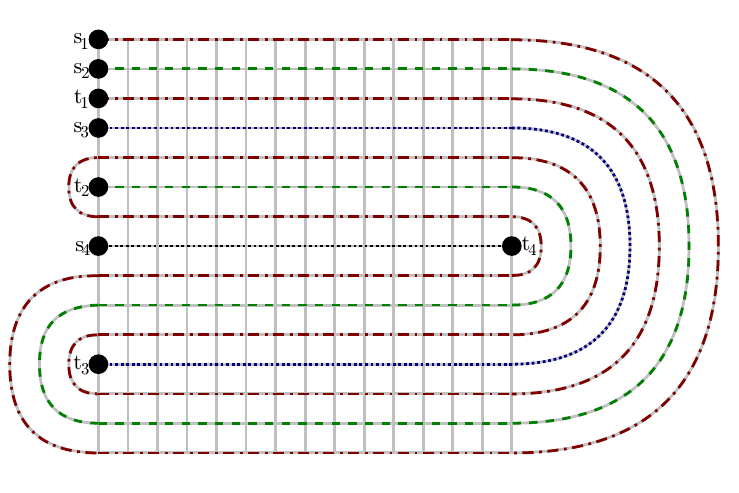}}
	\caption{$G_4, (s_1, t_1), (s_2, t_2), (s_3, t_3), (s_4, t_4)$ with solution.\label{ExGridS}}
\end{figure}

In our proof we construct a family of graphs $G_k, k \geq 1$, of tree-width and path-width $\geq 2^k - 1$, and with a vital linkage with $k$ components. Figure~\ref{ExGridS} shows the graph $G_4$.

\begin{figure}
	\centering
	\scalebox{1.7}{\includegraphics{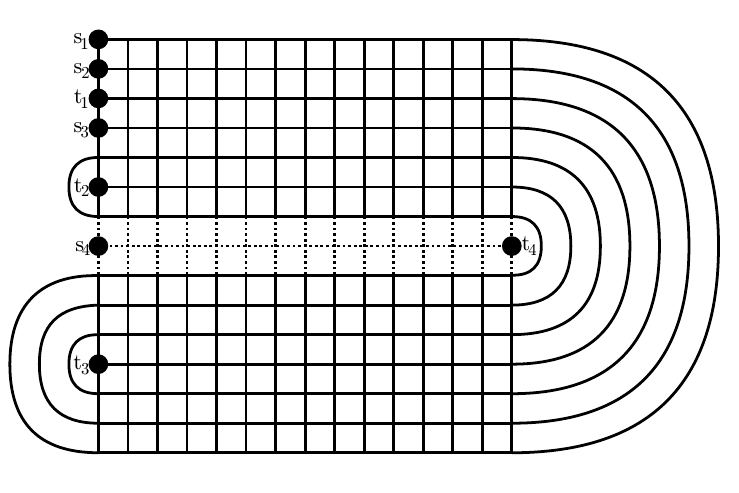}}
	\caption{The construction of $G_4 = G_{4, 15}$ from $G_{3, 30}$.\label{ExG4Def}}
\end{figure}

\begin{definition}[The graph $G_k$]
	Let $k, p \in \mathbb{N} \setminus \{0\}$. We inductively define an instance
	$G_k, (s_1, t_1), \ldots, (s_k, t_k)$ of $\DPP$ as follows.

	The Graph $G_{1, p}$ is the path $x_1, x_2, \ldots, x_p$ with $p$ vertices, $s_1(G_{1, p}) \coloneqq x_1$, $t_1(G_{1, p}) \coloneqq x_p$. The \emph{bottom row} and the \emph{top row} of $G_{1, p}$ are the graph $G_{1, p}$ itself.

        We define the graph $G_{k + 1, p}$ by adding a path $y_1, y_2, \ldots, y_p$ with $p$ vertices to $G_{k, 2p}$ as follows.  Let $x_1, x_2, \ldots, x_{2p}$ be the bottom row of $G_{k, 2p}$ and let $z_1,z_2, \ldots, z_{2p}$ be the top row of $G_{k, 2p}$. Let 
	\begin{align*}
		V(G_{k + 1, p}) \coloneqq V(G_{k, 2p}) &\cup \{y_1, y_2, \ldots, y_p\},\\
		E(G_{k + 1, p}) \coloneqq E(G_{k, 2p}) &\cup \left\{\{y_i, y_{i + 1}\}\ \middle|\ 1 \leq i  < p \right\} \cup\\
			&\left\{\{y_i, x_i\}, \{y_i, x_{2p - i + 1}\}\ \middle|\ 1 \leq i \leq p \right\}.
	\end{align*}
	We set $s_{k + 1}(G_{k + 1, p}) \coloneqq y_1$, $t_{k + 1}(G_{k + 1, p}) \coloneqq y_p$ and $s_i(G_{k + 1, p}) \coloneqq s_i(G_{k, p})$, $t_i(G_{k + 1, p}) \coloneqq t_i(G_{k, p})$ for $1 \leq i \leq k$. The \emph{top row} of $G_{k + 1, p}$ is $z_{1}, \ldots, z_{p}$ and the \emph{bottom row} of $G_{k + 1, p}$ is $z_{2p}, \ldots, z_{p + 1}$.

	Let $G_k \coloneqq G_{k, 2^k - 1}$. We define the $\DPP$ instance $G_k, (s_1, t_1), \ldots, (s_k, t_k)$ as $G_k, (s_1(G_k), t_1(G_k)), \ldots, (s_k(G_k), t_k(G_k))$.
\end{definition}

	Figure~\ref{ExG4Def} shows the construction of $G_4 = G_{4, 15}$ from $G_{3, 30}$.

\begin{remark}\label{rem:tree-width}
	By construction, the graph $G_{k}$ contains a $((2^k - 1) \times (2^k - 1))$-grid as a subgraph. The tree-width and path-width of $G_k$ are thus at least $2^k - 1$.
\end{remark}

\begin{remark}\label{rem:linkage}
	By construction, the graph $G_{k}$ contains a linkage (because in each step we add a path linking a new terminal pair).
\end{remark}

We will now show that this linkage is vital by considering a topological version.

\begin{definition}[Topological DPP]\label{dfn:top-DPP}
	Given a subset $X$ of the plane and $k$ pairs of terminals $(s_{1},t_{1})\in X^{2}, \ldots, (s_{k},t_{k})\in X^{2}$ the {\sc topological Disjoint Paths Problem} is the problem of deciding whether there are $k$ pairwise disjoint curves in $X$, such that each curve $P_i$ is homeomorphic to $[0, 1]$ and its ends are $s_i$ and $t_i$.
	If such curves $P_1, \ldots, P_k$ exist, we refer to them as a \emph{solution}.
	We denote an \emph{instance} of the topological Disjoint Paths Problem by $X, (s_1, t_1), (s_2, t_2), \ldots, (s_k, t_k)$.
\end{definition}

	A \emph{disc-with-edges} is a subset $X$ of the plane containing a closed disc $D$ such that the connected components of $X \setminus D$, called \emph{edges}, are homeomorphic to open intervals $(0, 1)$. We now define a family $(X_k)_{k \in \mathbb{N} \setminus \{0\}}$ of discs-with-edges together with terminals. These will be used as instances
	of the topological $\DPP$. Figure~\ref{Disc-with-edges} illustrates the construction.

\begin{figure}
	\centering
	\scalebox{1.7}{\includegraphics{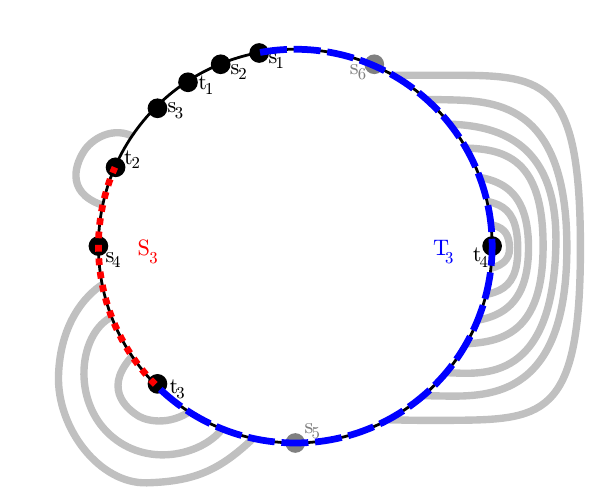}}
	\caption{The construction of $X_4, (s_1, t_1), \ldots, (s_4, t_4)$, for the topological $\DPP$.\label{Disc-with-edges} Note that $s_5$ and $s_6$ are only used to place $E_4$ correctly.}
\end{figure}

\begin{definition}[$X_k$]\label{dfn:Disc-with-edges}

	Let $D$ be a closed disc in the plane and $k \in \mathbb{N} \setminus \{0\}$.
	We start by inductively defining points $s_k$, $t_k$ on the boundary $\partial D$ of $D$. (These will be used as terminals and to confine the way the edges are added to $D$.)
	Let $s_1,t_1$ be two distinct points on $\partial D$, and let 
	$C_1\coloneqq \partial D\setminus\{s_1,t_1\}$. Hence $C_1$ is the union of two curves, each homeomorphic to the open interval $(0, 1)$. Call one of the curves $S_1$ and the other $T_1$.
	Assume that $s_k$, $t_k$, $C_k$, $S_k$, and $T_k$ are already defined, and assume that $T_k$ is a curve adjacent to $t_k$ and $s_1$.
	Place a new point $s_{k+1}$ on $S_k$ and a new point $t_{k+1}$ on $T_k$, let $C_{k+1}\coloneqq C_k\setminus\{s_{k+1},t_{k+1}\}$, let $T_{k+1}$ be the component of $C_{k+1}$ adjacent to $t_{k+1}$ and $s_1$, and let $S_{k+1}$ be the component of $C_{k+1}$ adjacent to $t_{k+1}$ and $t_k$.

%

	Now let $X_1 \coloneqq D$ and $E_1\coloneqq \emptyset$. Assume the space $X_k$ and the set $E_k$ are already defined. We define $X_{k+1}$ by adding a planar matching of $2^k - 1$ edges to $X_k$. We call the set of these edges $E_{k + 1}$. The edges are pairwise disjoint and disjoint from $X_k$. They are added such that each end is adjacent to a point on $\partial D$ and no two edges are adjacent to the same point on $\partial D$. Each edge has one end adjacent to a point on the component of $C_{k + 2}$ between $t_k$ and $s_{k + 1}$, and the other end adjacent to a point on the component of $C_{k + 2}$ between $t_k$ and $s_{k + 2}$. Finally, let $X_{k + 1}\coloneqq X_k\cup E_{k + 1}$.

	In this way we obtain a family $X_k$, $(s_1, t_1), (s_2, t_2), \ldots, (s_k, t_k)$ of instances to the topological DPP.
\end{definition}

\begin{figure}
	\centering
	\scalebox{1.7}{\includegraphics{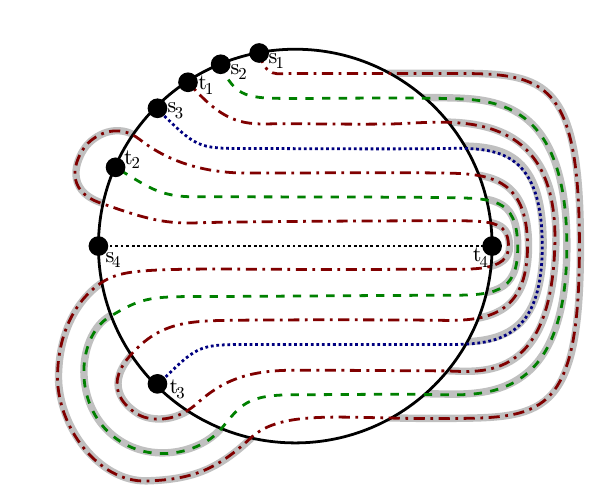}}
	\caption{A solution of the topological $\DPP$ from Figure~\ref{Disc-with-edges}.\label{ProofG}}
\end{figure}

\begin{remark}\label{rem:top-linkage}
	The embedding of $G_{k}$ (as shown in Figure~\ref{ExGridS} for $G_4$) corresponds to the space $X_k$. Thus by Remark~\ref{rem:linkage} the topological $\DPP$ on $X_k$, $(s_1, t_1), (s_2, t_2), \ldots,$ $(s_k, t_k)$ has a solution.
\end{remark}

For an instance of the topological $\DPP$ on $X_4$, this solution can be seen in Figure~\ref{ProofG}.

\begin{lemma}\label{lem:uniqueness}
       For $k \in \mathbb{N} \setminus \{0\}$ the topological $\DPP$ instance $X_k, (s_1, t_1), \ldots, (s_k, t_k)$ has a unique solution $P_1, \ldots, P_k$ (up to homeomorphism). The solution uses all edges $\bigcup_{1\leq i \leq k}E_i$.
\end{lemma}

\begin{figure}
	\centering
	\scalebox{1.7}{\includegraphics{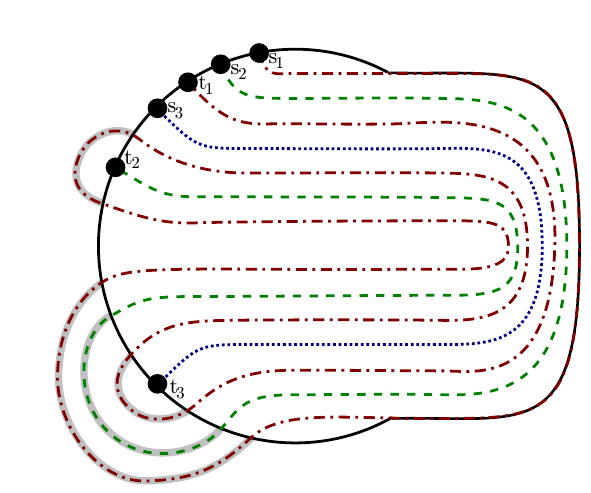}}
	\caption{The solution on $X'_3, (s_1, t_1), \ldots, (s_3, t_3)$ induced by the solution on $X_4, (s_1, t_1), \ldots, (s_4, t_4)$.\label{ProofG-ind}}
\end{figure}

\begin{proof}
	For $k = 1$ this is true because $E_1 = \emptyset$. Inductively assume that the lemma holds for $k$.
Let $ P_1, \ldots,  P_{k + 1}$ be any solution to  $X_{k +
1}, (s_1, t_1), \ldots$, $(s_{k + 1}, t_{k + 1})$. This solution induces a
solution of the topological $\DPP$ $X_k, (s_1, t_1)$, $\ldots$, $(s_k, t_k)$ as
follows. Every edge $e \in E_{k + 1}$ together with the segment of $\partial
D$ that connects the ends of $e$ and contains $t_{k + 1}$ bounds a disc $D_e$.
	The space $X'_k\coloneqq X_{k+1} \cup \bigcup_{e \in E_{k + 1}}D_e$ is
homeomorphic to $X_k$ and the paths $ P_1, \ldots,  P_k$ 
form a solution of $X'_k,$ $(s_1, t_1)$, $\ldots$, $(s_k, t_k)$. 
Figure~\ref{ProofG-ind} illustrates this for $k = 3$.
%
By induction, this solution is unique up to homeomorphism 
and the paths $ P_1, \ldots,  P_k$ use all edges in 
	$\bigcup_{1\leq i\leq k}E_i$. Let $Q_1,\ldots, Q_k$ be the solution
obtained by embedding the graph $G_{k}$ (cf.~Remark~\ref{rem:top-linkage}).
By uniqueness, for each $i\in[k]$, the edges of $\bigcup_{1\leq
i\leq k}E_i$ used by $P_i$ are the same as for $ Q_i$, 
and the order of their
appearance on $P_i$ when walking from $s_i$ to $t_i$ is also the same as on $Q_i$.
Hence the solution $ P_1, \ldots,  P_k$ on $X'_k$ restricted to the closed disc $D$ of $X'_k$ 
is a planar matching of curves (the curves in $\bigcup_{1 \leq i \leq k} P_i \setminus
	\bigcup_{1 \leq i \leq k} E_i$) between pairs of points on $\partial D$ (and the same pairs
	of points are obtained by restricting $Q_1, \ldots, Q_k$ to $D$). These pairs of points also have to
be matched in $X_{k+1}$.

We now claim that in the solution $ P_1, \ldots,  P_{k + 1}$ on $X_{k+1}$, 
	each curve in $\bigcup_{1 \leq i \leq k} P_i \setminus
	\bigcup_{1 \leq i \leq k} E_i$ uses an edge of $E_{k + 1}$. 
	If not, then there is a curve \[p\in\bigcup_{1 \leq i \leq k} P_i \setminus \bigcup_{1 \leq i \leq k} E_i\] that avoids all edges in $E_{k + 1}$.
Since the edges of $\bigcup_{1\leq i\leq k}E_i$ are already used, $p$ is
routed within $D$. By construction of $X_{k+1}$ and the fact that all
edges of $\bigcup_{1\leq i\leq k}E_i$ are already used, 
this means that $p$ separates $s_{k + 1}$ from both $t_{k + 1}$ and 
the endpoints of the edges in $E_{k+1}$, 
a contradiction to $ P_{k + 1}$ being a path in the 
solution. Hence $p$ uses an edge of $E_{k + 1}$.

Since the sets $\bigcup_{1 \leq i \leq k} P_i \setminus \bigcup_{1 \leq i \leq
k} E_i$ and $E_{k + 1}$ have equal size, it follows that
	each curve of the matching \[\bigcup_{1 \leq i \leq k} P_i \setminus
	\bigcup_{1 \leq i \leq k} E_i\] uses precisely one edge of $E_{k + 1}$.
Since the endpoints of the matching are fixed, they induce an order on the matching
curves which determines precisely which edge of $E_{k + 1}$ is
used by which curve. 

	Altogether, this shows that the solution to 
$X_{k +1}, (s_1, t_1), \ldots, (s_{k + 1}, t_{k + 1})$ is unique up to homeomorphism and
	uses all edges $\bigcup_{1\leq i \leq k}E_i$.

\end{proof}

\begin{figure}
	\centering
	\scalebox{1.7}{\includegraphics{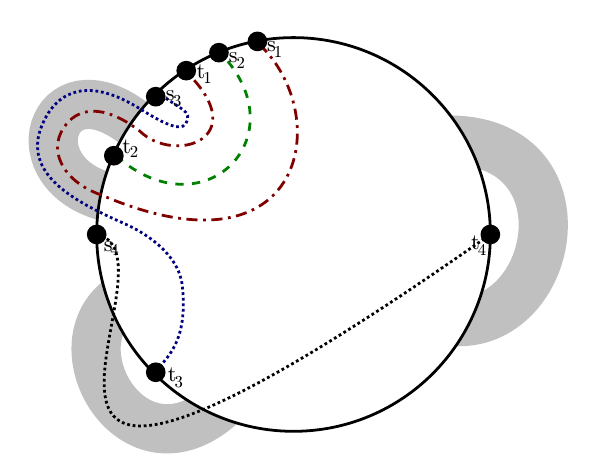}}
	\caption{The number of edges around the terminals is crucial (cf.\ Remark~\ref{rem:no-of-edges}).\label{NoProof}}
\end{figure}

\begin{remark}\label{rem:no-of-edges}
	In a topological $\DPP$ instance, the number of edges around the terminals is crucial. 
	Even just relaxing the conditions on $X_k$ by having 2 edges instead of 1 edge around terminal $t_2$ allows a quite different solution to the topological $\DPP$.
	This solution uses no edge around $t_k$, one edge around each of $t_3, t_3, \ldots, t_{k - 1}$, and the two edges around $t_2$ (Figure~\ref{NoProof}
	shows this for $k=4$).
\end{remark}

\begin{theorem}\label{theo:main}
	Let $k \in \mathbb{N} \smallsetminus \{0\}$. The graph $G_k$ contains a vital linkage.
\end{theorem}

\begin{proof}
	Let $P_1, \ldots, P_k$ be the linkage from Remark~\ref{rem:linkage}. We argue that it is vital.
	For $k = 1$ and $k = 2$, one can easily verify that $G_k$ has a unique embedding. For $k \geq 2$, contracting an edge at $s_1$ suffices to make $G_k$ $3$-connected. Since 3-connected planar graphs have unique embeddings~\cite{Whitney1932}, the graph $G_k$ also has a unique embedding, and it suffices to consider our previous embedding of $G_k$ (cf.~Figure~\ref{ExGridS}). Let $D$ be the minimal disc containing the grid in $G_k$. The disc $D$ together with $E(G_k)$ is the space $X_k$. The paths $P_1, \ldots, P_k$ thus give a solution to the topological $\DPP$ instance $X_k, (s_1, t_1), \ldots, (s_k, t_k)$, which by Lemma~\ref{lem:uniqueness} is unique and uses all edges in $E_k$.
	Thus any linkage $P_1', \ldots, P_k'$ with the same pattern as $P_1, \ldots, P_k$ can differ from $P_1, \ldots, P_k$ only inside the grid. Thus for each $y \in [2^k - 1]$ there is a subpath $Q_y'$ of some path of the solution $P_1', \ldots, P_k'$, such that the endpoints of $Q_y'$ are $(y', 1)$ and $(y, 2^k - 1)$ for some $y' \in [2^k - 1]$. Hence the family $(Q'_y)_{y \in [2^k - 1]}$ is a
	linkage between the first column and the last column of the grid.

	Suppose that $P_1', \ldots, P_k'$ indeed differs from $P_1, \ldots, P_k$. 
	Then at least one path $Q_y'$ contains a vertical edge $e$ in the grid.
	Hence the column of $e$ contains at most $2^k-3$ vertices that are not used by $Q_y'$ and, by Menger's Theorem~\cite{Menger1927}, the remaining $2^k-2$ paths of the family cannot be routed, a contradiction.

\end{proof}

\emph{Proof of Theorem~\ref{thm:lowerbound}}
Theorem~\ref{thm:lowerbound} immediately follows from Theorem~\ref{theo:main} and Remark~\ref{rem:tree-width}.
\qed

\section*{Acknowledgements}
\small{This research was partially supported by the Deutsche Forschungsgemeinschaft, project Graphstrukturtheorie und algorithmische Anwendungen, AD 411/1-1 and project Graphstrukturtheorie im {\"U}bersetzerbau, KR 4970/1-1.
We thank Fr{\'e}d{\'e}ric Mazoit for valuable discussions, especially for inspiring Remark~\ref{rem:no-of-edges}.
We also thank an anonymous reviewer for suggesting very elegant shortenings of our construction and proof.}

\vspace{3mm}

\bibliographystyle{plain}
\bibliography{Lower}

\end{document}